\documentclass[12pt, onecolumn, draftclsnofoot]{IEEEtran}
\usepackage{bbding}
\usepackage{url}
\usepackage{booktabs}
\usepackage{subfigure}
\usepackage{multirow}
\usepackage{fancybox}
\usepackage[dvipsnames,svgnames,x11names]{xcolor}
\usepackage{graphicx}
\usepackage{bbm}
\usepackage{cite}
\usepackage{epsfig, psfrag, amsmath, amssymb, amsfonts, latexsym, eucal, balance, hyperref, indentfirst,bookmark, amsthm}
\newtheorem{thm}{Theorem}%[section]
\newtheorem{cor}{Corollary}

\makeatletter

\newcommand{\Rmnum}[1]{\expandafter\@slowromancap\romannumeral #1@}
\makeatother

\graphicspath{{figs/}}
\begin{document}

\title{Tradeoff between Delay and Physical Layer Security in Wireless Networks}
\author{Yi Zhong, Xiaohu Ge, Tao Han, Qiang Li, Jing Zhang
\thanks{
The authors are with
School of Electronic Information and Communications, Huazhong University
of Science and Technology, Wuhan, P. R. China. (e-mail: yzhong@hust.edu.cn).
This research was supported by the National Natural Science Foundation of China (NSFC)
grant No. 61701183.

Part of this work was submitted to the 2018 IEEE International Conference on Communications (ICC).
}}
\maketitle
\pagestyle{empty}  % no page number for the second and the later pages
\thispagestyle{empty} % no page number for the first page
\begin{abstract}
Exchange of crucial and confidential information leads to the unprecedented attention on the security problem in wireless networks.
Though the security has
been studied in a number of works, the joint optimization of the physical layer security
and the end-to-end delay management, which requires
a meticulous cross-layer design, has seldom been evaluated.
In this work, by combining the tools from stochastic geometry and queueing theory, we analyze the tradeoff between the delay and the security performance in large wireless networks. We further propose a simple transmission mechanism which splits a message into two packets and evaluate its effect on the mean delay and the secrecy outage probability.
Our numerical results reveal that the security performance is better for larger path loss exponent when the density of legitimate nodes is large, and it is reverse when the density is small.
Moreover, it is observed that by introducing the simple mechanism of message split, the security performance is greatly improved in the backlogged scenario and slightly improved in the dynamic scenario when the density of legitimate transmitters is large.
In summary, this work provides an understanding and a rule-of-thumb for the practical
design of wireless networks where both the delay and the security are key concerns.
\end{abstract}

\newpage

\section{Introduction}
\label{sec:introduction}
\subsection{Motivations}
The wireless access becomes ubiquitous for the civilian and military applications \cite{ge20165g}, including the exchange of crucial and confidential information like the banking related data, leading to unprecedented attention to the security problem for the exchange of information between nodes in the wireless networks. Comparing with the traditional security methods based on the cryptographic techniques at the upper layers
of the communication protocols in a wireless network \cite{4785384}, the physical layer security is an information-theoretic based approach that guarantees the security of wireless links by exploiting the randomness of wireless channels, which makes it more difficult for the attackers to decipher the delivered messages. From the point of view of the information theory, the security problem can be modeled by the wiretap channel proposed by Wyner in his pioneering work \cite{6772207}, where a source sends the information to a legitimate receiver at the presence of an eavesdropper.

On the other hand, due to the emergence of new types of latency-critical applications such as the command-and-control of drones, the advanced manufacturing, the vehicular networks \cite{7981531} and the tactile Internet, the delay performance becomes an indispensable metric that should be considered in the wireless transmissions \cite{6824752}. Though the delay in wireless networks has been studied by a number of works \cite{7886285,yi2007commag}, the joint optimization of the physical layer security and the end-to-end delay management has seldom been studied in the literature, which requires a meticulous cross-layer design.
The relationship between the physical layer security and the end-to-end delay is exceedingly complex, which is affected by a number of factors, such as the random time-varying channel, the irregular deployment of nodes, the wireless transmission mechanisms and so on.

\begin{figure}[!ht]
\centering
\includegraphics[width=0.7\textwidth]{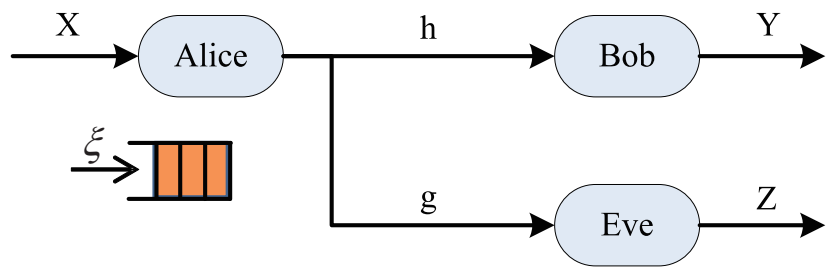}
\caption{Cipher system with random arrival messages: Alice attempts to deliver messages to Bob, while keeping it secret from Eve who intercepts the transmissions.}
\label{fig:Figure1}
\end{figure}

An intuitive tradeoff between the delay and the physical layer security could be observed from Figure \ref{fig:Figure1}, where a transmitter named Alice attempts to deliver messages to a legitimate receiver named Bob, while keeping it secret from an eavesdropper named Eve who intercepts the transmissions.
Alice maintains a buffer of infinite capacity to store the incoming
packets, and the time is divided into discrete slots with equal duration. In each time slot, Alice attempts to
transmit its head-of-line packet with a certain probability $p$ if its buffer is not empty.
A failure occurs if Bob cannot decode the packets sent by Alice. When a failure occurs, Alice retransmits the packet in the next time slot with the same access probability $p$.
The security constraint is that the probability of a packet delivered by Alice failing to achieve perfect secrecy from the detrimental eavesdropper Eve is less than certain threshold.
Increasing the transmit power of Alice reduces the delay but makes it more easier for Eve to decode the message, thus decreasing the security as well.
Moreover, in wireless networks with multiple legitimate links and multiple eavesdroppers, the tradeoff between the delay and the physical layer security is much more perplexed. For example, increasing the transmit probability $p$ may decrease the delay for legitimate links and increase the chance for Eve to intercept the messages since the probability to schedule the packets is increased, but it also increases the interference in wireless networks which makes it harder for Eve to intercept the messages.

The motivation of this work is to evaluate the tradeoff between the delay and the physical layer security in wireless networks.
Particularly, we model the spatial distribution of legitimate nodes and eavesdroppers in wireless networks by using the tools from the point process theory.
Through combining the point process theory and the queueing theory, we analyze the tradeoff between the delay and the physical layer security theoretically.
Our work will provide a useful guideline for the design of secure wireless transmissions while meeting the end-to-end delay requirements imposed by the mobile applications in wireless networks.

\subsection{Related Works}
The work in \cite{1055892} extends the Wyner's wiretap channel model \cite{6772207} to the case of broadcast channels with confidential messages, which shows that perfect secrecy could be achieved if the legitimate link has a better channel than the link of eavesdropper.
The work in \cite{1055917} demonstrates that the secrecy capacity, i.e., the maximum achievable rate of the legitimate link while guaranteeing that the eavesdropper cannot decode the message, is the difference
of the maximum achievable rates between the legitimate link and the eavesdropping link.

Most of the studies on physical layer security in the literature focused on the scenes of a small number of nodes.
For instance, the works in \cite{Shafiee2007Achievable,5485016} studied the security problem in the case of a small number of nodes with multiple antennas, and the works in \cite{4608977, 4455769, 4529293} studied the security problem in the cooperative transmission and in the multiple access channels.
The works in \cite{7470273,7926385} studied the secure and reliable transmission for the multiple-input multiple-output systems.
Recently, due to the applications of stochastic geometry tools in the wireless networks \cite{haenggi2009stochastic}, a few studies have been carried out to analyze the security problem in large wireless networks. For example, the works in
\cite{4595044,5759739,5649155,5648780} studied the
connectivity problem in the wireless network where the legitimate nodes and the eavesdroppers are randomly distributed in space, and the works in \cite{5757463,6142080,6034725} evaluated the coverage and the capacity with the limitation of security.
The work in \cite{6512533} studied the information-theoretic secrecy performance in large cellular networks by modeling the locations of both base stations and mobile users as two independent Poisson point processes (PPP) \cite{haenggi2012stochastic,baccelli2009stochastic2,andrews2010tractable,6781590,8070348}.
The work in \cite{5934342} studies the throughput of large decentralized wireless networks with physical layer security constraints, and the authors also utilize the tools from stochastic geometry to model the wireless networks.

Though the relationship between the security and the capacity of large wireless networks has been well explored in the literature, few works have been done to evaluate the relationship between the security and the delay in large wireless networks. Part of the reason is that the analysis of delay in large wireless network is
far more difficult than the analysis of coverage and capacity since the delay is a complicated function of all links and is
affected by an unusually large number of variables such as the network load, the medium access protocol, the path loss, and so on.
In \cite{7467553}, the joint optimization of the physical layer security and the end-to-end delay is studied in the wireless body area networks.
In \cite{7012108}, the resource allocation problem is studied in the cognitive radio network with delay
and security constraints.
However, these works considers few nodes in the network and have not obtained the analytical results for the large wireless networks.
Recently, several tentative works appear to analyze the delay in large scale wireless networks, such as \cite{blaszczyszyn2015performance}.
The basic idea is to combine the stochastic geometry and the queueing theory, which results in the challenge of interacting queues problem \cite{rao1988stability,luo1999stability,ephremides1987delay}, i.e., the service rates of the queues in the networks rely on the status of all queues in the networks.
Our previous works have derived the conditions for the stability of queues \cite{7486114} and the bounds for the delay \cite{7886285} in large wireless networks.
%In this paper, we evaluate the tradeoff between the delay and the security in large wireless networks.

\subsection{Contributions}
In this work, we evaluate the tradeoff between the delay and the security performance in large wireless networks for both the backlogged scenario and the dynamic scenario.
To our best knowledge, this is the first work to explore the joint optimization of the delay and the security problem in large wireless networks.
By combining the tools from stochastic geometry and queueing theory, we introduce some new analytical approaches and derive the close-formed results for the mean delay and the secrecy outage probability (i.e., the probability of a transmission failing to achieve perfect secrecy). We also propose a simple transmission mechanism which splits a message into two packets and analyze its effect on the delay and security performance.

Our results reveal that the security performance of a large wireless network is better for larger path loss exponent when the density of legitimate nodes is large, and it is reverse when the density is small.
We also find that under the condition that a certain confidential rate (i.e., the achievable rate of confidential messages with the constraint of perfect secrecy) is guaranteed, the delay performance is better in high signal-to-interference ratio (SIR) regime than that in low SIR regime.
By introducing the simple mechanism of message split, it is shown that the delay performance in the backlogged scenario is improved when the density of legitimate transmitters is large.
Meanwhile, introducing message split greatly improves the
security performance in the backlogged scenario, while slightly improves the security performance in the dynamic scenario for small density of legitimate nodes.

To sum up, our contributions in this work could be summarized as follows.
\begin{itemize}
\item The joint optimization of the delay and the
security problem in large wireless networks, which is seldom explored in the literature, is evaluated through numerical analysis.
\item New analytical approaches based on the combination of stochastic geometry and queueing theory are proposed to derive the mean delay and the secrecy outage probability in the scenario of random traffic.
\item Numerical evaluation based on the theoretical analysis is conducted to gain insight for the  practical design of wireless networks where both the delay and the security are key concerns.
\end{itemize}

The remaining part of the paper is organized as follows.
Section \ref{sec:system} describes the propagation model, the traffic model, and the secrecy transmission rate. Section \ref{sec:delay} analyzes the delay in two different scenarios respectively, i.e., the backlogged scenario and the dynamic scenario.
Section \ref{sec:perform} proposes and analyzes a simple transmission mechanism, in which a confidential message is divided into two packets to be delivered independently.
Section \ref{sec:numerical} numerically analyzes the tradeoff between the security and the delay performance.
Finally, Section \ref{sec:conclusions} concludes the paper.

\section{System Model}
\label{sec:system}
\begin{figure}
\centering
\includegraphics[width=0.65\textwidth]{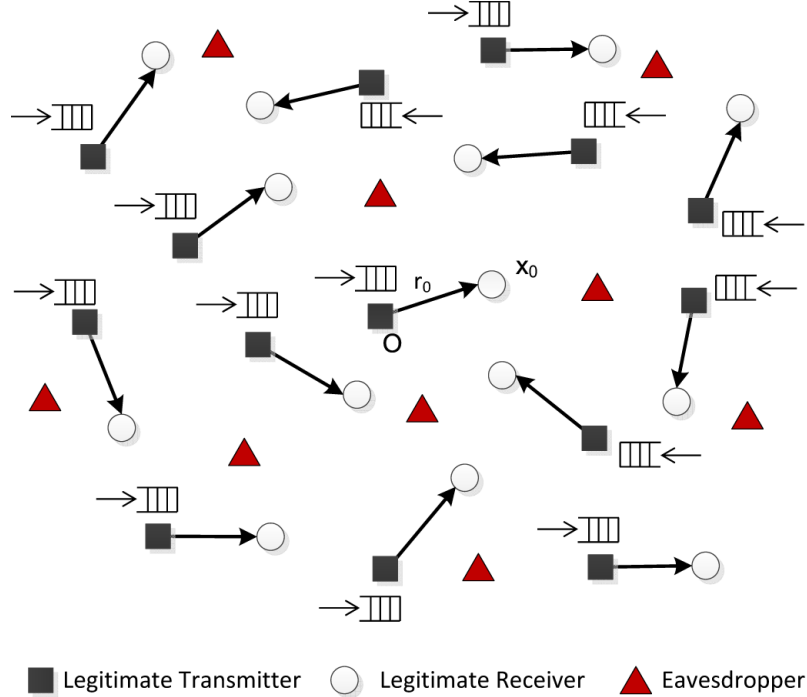}
\caption{Topology of the network, where the squares denotes the legitimate transmitters, the circles denote the legitimate receivers and the triangles denote the eavesdroppers.}
\label{fig:bipolar}
\end{figure}

We consider an ad hoc network consisting of both legitimate nodes and eavesdroppers in a large two-dimensional space.
The locations of the legitimate transmitters are modeled as a homogeneous
PPP $\Phi_l$ with intensity $\lambda_l$, and the locations of the eavesdroppers
are modeled as another independent homogeneous PPP $\Phi_e$ with intensity $\lambda_e$.
Each legitimate transmitter is paired with a legitimate receiver
at a fixed distance $r_0$ and a random orientation.
Without loss of generality, we consider a typical legitimate transmitter located at the origin, and the corresponding legitimate receiver is located at $x_0$ with $|x_0|=r_0$ (see Fig. \ref{fig:bipolar}).
We consider the discrete time system where the time is divided into discrete slots with equal duration, and each transmission attempt occupies exactly one time
slot.
In practice, the locations of the transmitters and the receivers are generated once at the beginning and then kept unchanged during all time slots.
Therefore, the topology of the network could be considered as static.

\subsection{Propagation Model}
The propagation model in this work consists of two parts, i.e., the path loss and the fading.
We assume that the standard power path loss model is used, and the average received power at a receiver $y$ from a transmitter $x$ is denoted by
\begin{equation}
P_r=P_t|x-y|^{-\alpha},
\end{equation}
where $P_t$ is the transmit power, $P_r$ is the received power, and $\alpha$ is the path loss exponent ($\alpha\geq2$).
The fading of each link is assumed to be Rayleigh block fading with zero mean and unit variance, and it is independent between different links. The
power fading coefficients in different time slots are independent identically distributed (i.i.d.) and are constant during one time slot.
Note that the channel model in our work is not as comprehensive as that in \cite{7913700} which considers the path loss, the Rayleigh fading, the azimuth angle of departure and the log-normal shadowing.
However, the channel model in our work is a common assumption that has been widely accepted in the literatures, such as \cite{net:Haenggi13tit}. Meanwhile, using the model of the path loss and the Rayleigh block fading may greatly reduce the complexity of the derivations and make the analysis tractable, and close-formed results could be obtain to gain intuitive insight.
Due to the fact that most of the existing wireless networks are interference-limited, i.e., the noise is much smaller compared with the interference, we ignored the thermal noise in our analysis.
Without loss of generality, the transmit power of all legitimate transmitters is normalized.

When all transmitters are active, the SIR at the typical legitimate receiver $x_0$ is
\begin{equation}
\mathrm{SIR}_{x_0}=\frac{h_{x_0}r_0^{-\alpha}}{\sum_{y\in\Phi_l}h_y|y-x_0|^{-\alpha}},
\end{equation}
where $h_{x_0}$ is the fading coefficient between the typical legitimate transmitter at the origin and the corresponding receiver at $x_0$, $\{h_y\}$ is the fading coefficients between the transmitter $y\in\Phi_l$ and the typical legitimate receiver $x_0$.

When all transmitters are active, the SIR at an eavesdropper $x_e\in\Phi_e$ is
\begin{equation}
\mathrm{SIR}_{x_e}=\frac{h_e|x_e|^{-\alpha}}{\sum_{y\in\Phi_l}h_y|y-x_e|^{-\alpha}},
\end{equation}
where $h_e$ is the fading coefficient between the legitimate transmitter at the origin and the eavesdropper at $x_e$.

\subsection{Traffic Model}
For completeness consideration, we consider two scenarios for the traffic arrival process, i.e., \emph{the backlogged scenario} and \emph{the dynamic scenario} (see Figure \ref{fig:twomodes}).

\begin{figure}
\centering
\includegraphics[width=0.65\textwidth]{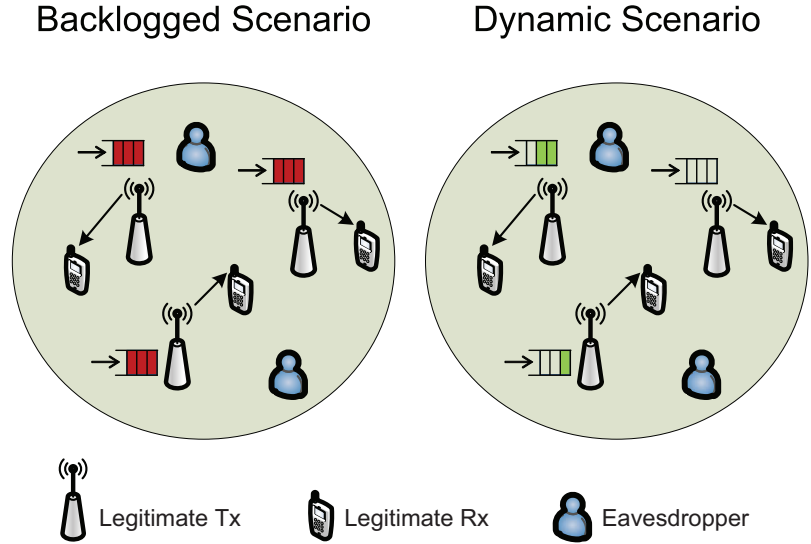}
\caption{Two scenarios for the traffic arrival process, i.e., the backlogged scenario and the dynamic scenario.}
\label{fig:twomodes}
\end{figure}

\subsubsection{Backlogged scenario}
In the backlogged scenario, we assume that the network is backlogged, i.e., the transmitters always have packets to transmit when they are scheduled for transmission. In this way, a meaningful and practically relevant metric in the backlogged scenario is the duration to successfully deliver one packet, which is closely related to the number of retransmissions of a packet. This type
of delay, which ignores the queueing delay (i.e., the waiting time of a packet until it is served), is also defined as local delay \cite{5462132,net:Haenggi13tit,zhong2014managing}.
Without loss of generality, we call it the delay for simplicity in the backlogged scenario.

\subsubsection{Dynamic scenario}
In the dynamic scenario, we assume that the packets arrive at the legitimate transmitters as a stochastic process.
The packets arrival process at each legitimate transmitter $x_i\in\Phi_l$ is assumed to be an independent Bernoulli process of arrival rate $\xi$ per time slot. By the definition of the Bernoulli process, $\xi$ is the probability of an arrival of a packet at user $x_i$ in any time slot.

We assume that each legitimate transmitter maintains a buffer with infinite capacity to store the incoming packets.
In each time slot, each legitimate transmitter attempts to transmit its head-of-line packet with probability $p$ if the queue at the transmitter is not empty.
As for the retransmission mechanism, we assume that if a packet fails for transmission in the scheduled time slot,
the packet will be added into the head of the queue of the corresponding legitimate transmitter and wait to be rescheduled again.

\subsection{Secrecy Transmission Rate}
According to Wyner¡¯s secure encoding scheme \cite{6772207,4802331}, a transmitter may choose
two rate thresholds, i.e., the rate threshold of codewords $R_t$
and the rate threshold of confidential messages $R_s$.
The difference between the two rate thresholds, denoted by $R_e=R_t-R_s$, is the rate cost of securing the
messages against eavesdropper.
If the rate of a legitimate link is less than the threshold $R_t$, the legitimate receiver cannot decode the message with an arbitrarily small error.
Meanwhile, if the rate of every eavesdropping link is less than the threshold $R_e$, perfect secrecy could be  achieved to guarantee the confidential rate $R_s$.

With the above two rate thresholds, we define two failure events for a transmission attempt, i.e., the connection failure and the secrecy failure (see Figure \ref{fig:twofailures}).

\begin{figure}
\centering
\includegraphics[width=0.65\textwidth]{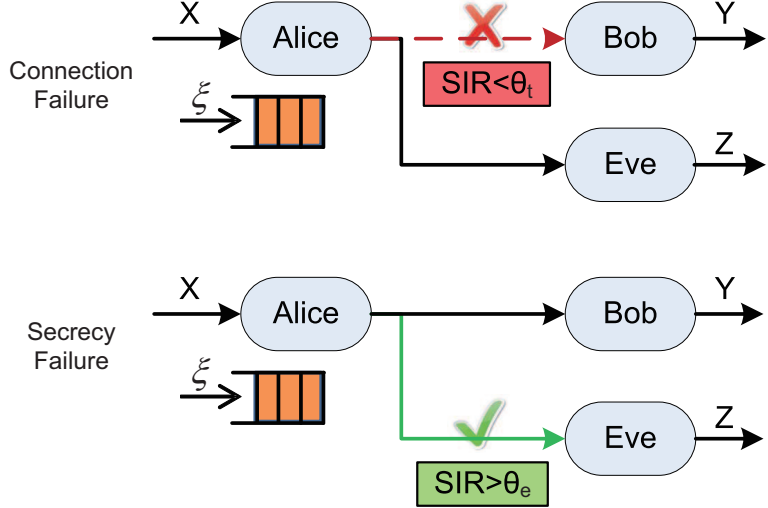}
\caption{Two failure events for a transmission attempt, i.e., the connection failure and the secrecy failure.}
\label{fig:twofailures}
\end{figure}

\subsubsection{Connection Failure}
The connection failure means that a transmission attempt of a packet sent by a legitimate transmitter cannot be successfully decoded by the corresponding receiver. The probability of connection failure for a transmission attempt is
\begin{eqnarray}
\mathcal{P}_{\rm cf}=\mathbb{P}\{\log_2(1+\mathrm{SIR}_{x_0})<R_t\}=\mathbb{P}\{\mathrm{SIR}_{x_0}<2^{R_t}-1\}. \end{eqnarray}
Let $\theta_t=2^{R_t}-1$ be the SIR threshold for the connection failure event.
Due to the retransmission mechanism, if a transmission attempt of a packet fails, the packet will be added into the head of the queue and wait to be rescheduled for transmission.

\subsubsection{Secrecy Failure}
The secrecy failure means that a transmission attempt of a packet sent by a legitimate transmitter is not perfectly secure against the eavesdroppers.
The probability of secrecy failure for a transmission attempt is
\begin{eqnarray}
\mathcal{P}_{\rm sf}=1-\prod_{x_e\in\Phi_e}\mathbb{P}\{\log_2(1+\mathrm{SIR}_{x_e})<R_e\}=1-\prod_{x_e\in\Phi_e}\mathbb{P}\{\mathrm{SIR}_{x_e}<2^{R_e}-1\}. \end{eqnarray}
Let $\theta_e=2^{R_e}-1$ be the SIR threshold for the secrecy failure event.
Since a packet may be retransmitted for several times due to the connection failure, the delivery of a packet is perfectly secure only when all the transmission attempts are successful.

Therefore, the delay in the backlogged scenario is the number of time slots required by a legitimate transmitter to successfully deliver a packet to the corresponding legitimate receiver. On the other hand, the delay in the dynamic scenario is the number of time slots between the arrival of a packet and the successful delivery of that packet. Meanwhile, the secrecy performance can be characterized by the secrecy outage probability, which in this paper is defined as the probability that at least one transmission attempt of a packet is not perfectly secure.
Let $N_r$ be the total number of scheduled transmissions until a packet is successfully decoded at the legitimate receiver, and $N_r$ is a geometric distributed random variable with success probability $1-\mathcal{P}_{\rm cf}$, i.e., $N_r\sim \mathrm{Geo}(1-\mathcal{P}_{\rm cf})$. Then, the secrecy outage probability with given $N_r$ is
\begin{eqnarray}
\mathcal{P}_{\rm so}=1-(1-\mathcal{P}_{\rm sf})^{N_r}. \label{eqn:soprob}
\end{eqnarray}

The relationship between the mean delay and the secrecy outage probability is not straightforward.
Intuitively, the equation (\ref{eqn:soprob}) reveals that increasing the total number of transmissions $N_r$ raises the mean delay as well as the secrecy outage probability.
However, since the probability of secrecy failure $\mathcal{P}_{\rm sf}$ may decrease, the secrecy outage probability may also decrease.
Therefore, the theoretical analysis of the mean delay and the secrecy outage probability is necessary.
In the following discussions, we derive the expressions to reveal the exact relationship between the mean delay and the secrecy outage probability.

The system model described above is a general mathematical framework that could be applied to analyze the security and delay performance of general large wireless networks, corresponding to a number of application scenarios. For example, the proposed mathematical framework could be applied in the military communications, where the soldiers from one army want to eavesdrop the confidential messages from the hostile army. In this application scenario, the soldiers from the eavesdropping army could be considered as eavesdroppers and the soldiers from the other army could be considered as the legitimate nodes. Another application scenario could be the prevention of theft in the wireless banking transactions, where the thieves could be considered as the eavesdroppers and the banking users could be considered as the legitimate nodes.
In these application scenarios, all eavesdroppers try to eavesdrop the confidential messages from all legitimate users. Thus, the legitimate users should prevent their confidential messages from being eavesdropped by any one of the eavesdroppers.

\section{Delay and Secrecy Outage Probability}
\label{sec:delay}
The delay in the backlogged scenario is exactly the transmission delay caused by retransmission (i.e., the time to successfully transmit a packet that is served), while the delay in the dynamic scenario consists of two parts, i.e., the queueing delay and the transmission delay.
We analyze the delay in the two different scenarios respectively in the following discussions.

\subsection{Backlogged Scenario}
In the backlogged scenario, whenever a transmitter is scheduled to access the channel it always has packet to transmit.
The backlogged assumption is widely used in the analysis of  wireless communication systems. With the backlogged assumption, the obtained interference in the network is an upper bound for the interference without backlogged assumption. Therefore, the performance derived with the backlogged assumption will be a useful bound for the performance of the practical system. With the backlogged assumption, considering the queueing delay is meaningless since the queueing delay is always infinite. However, a meaningful and practically relevant metric in the backlogged scenario is the transmission delay. The transmission delay in the backlogged scenario is exactly the local delay, which has been well explored in the literature, such as \cite{net:Haenggi13tit} and \cite{zhong2014managing}.
Another benefit of introducing the backlogged assumption is that the theoretical analysis of delay becomes tractable since the interacting between different queues is avoided. Concise and close-formed results could be obtained by the backlogged assumption, which further leads to more insights.
From \cite{net:Haenggi13tit}, we get the following lemma, which gives the analytical expression for the mean delay of the typical legitimate link.
\begin{thm}
\label{thm:1}
The mean delay of the typical legitimate link in the backlogged scenario is
\begin{equation}
\mathcal{D} = \frac{1}{p}\exp\left({\lambda_l \pi r_0^2\theta_t^{\delta}C(\delta)p}{(1-p)^{\delta-1}}\right), \label{eqn:localdelay}
\end{equation}
where $\theta_t=2^{R_t}-1$, $\delta=2/\alpha$ and $C(\delta)=\Gamma(1+\delta)\Gamma(1-\delta)=1/{\mathrm{sinc}(\delta)}$.
\end{thm}
\begin{proof}
The proof is ignored in this paper for conciseness since it is similar to the proof in \cite[Lemma 2]{net:Haenggi13tit} and \cite[Theorem 3]{zhong2014managing}.
\end{proof}

The result in Theorem \ref{thm:1} is closed-form and can be easily evaluated, and it shows directly how the mean delay varies with the transmit probability $p$.
The equation (\ref{eqn:localdelay}) indicates that the mean delay goes to infinite both when $p$ decreases to zero and when $p$ increases to one since $\delta<1$ holds.
When the transmit probability $p$ decreases to zero, a packet is scheduled for transmission for very small probability, and when the transmit probability $p$ increases to one, the interference and the interference correlation may greatly deteriorate the delay performance.

In the following, we derive the secrecy outage probability for the backlogged scenario.
Note that $N_r$ is the total number of transmissions until a packet is successfully decoded at the legitimate receiver.
Then, Theorem \ref{thm:1} gives the analytical expression for the mean times of transmissions $\mathbb{E}\{N_r\}=\mathcal{D}$ in the backlogged scenario.
For the typical legitimate link, when the realization of $\Phi_l$ is given, the connection failure probability is determined by the fading and random access mechanism, which are independent among different time slots. Thus, the connection failure probability with given $\Phi_l$, denoted by $\mathcal{P}_{\rm cf}^{\Phi_l}$, is the same for all time slots.
Let $N_r^{\Phi_l}$ be the total number of transmissions until a packet is successfully decoded at the legitimate
receiver when $\Phi_l$ is given.
Then, $N_r^{\Phi_l}$ is a geometric distributed random variable with success probability $\mathcal{P}_{\rm cf}^{\Phi_l}$, i.e., $N_r^{\Phi_l}\sim \mathrm{Geo}(1-\mathcal{P}_{\rm cf}^{\Phi_l})$.
In order to guarantee the perfect secrecy, each of the $N_r^{\Phi_l}$ transmission attempts should be perfect secure; otherwise, the message will be eavesdropped.
Letting $\mathcal{P}_{\rm sf}^{\Phi_l,\Phi_e}$ be the secrecy failure probability with given $\Phi_l$ and $\Phi_e$, the secrecy outage probability is
\begin{eqnarray}
\mathcal{P}_{\rm so}&=&1-\mathbb{E}_{\Phi_l,\Phi_e}\big[(1-\mathcal{P}_{\rm sf}^{\Phi_l,\Phi_e})^{N_r^{\Phi_l}}\big]\nonumber\\
&=& 1-\mathbb{E}_{\Phi_l,\Phi_e}\Big[\sum_{n=1}^\infty(1-\mathcal{P}_{\rm sf}^{\Phi_l,\Phi_e})^n\mathbb{P}\{N_r^{\Phi_l}=n\}\Big] \nonumber\\
&\stackrel{(a)}{=}&1-\mathbb{E}_{\Phi_l,\Phi_e}\Big[\sum_{n=1}^\infty(1-\mathcal{P}_{\rm sf}^{\Phi_l,\Phi_e})^n(1-\mathcal{P}_{\rm cf}^{\Phi_l})(\mathcal{P}_{\rm cf}^{\Phi_l})^{n-1}\Big] \nonumber\\
&=&1-\mathbb{E}_{\Phi_l,\Phi_e}\Big[\frac{1-\mathcal{P}_{\rm cf}^{\Phi_l}}{\mathcal{P}_{\rm cf}^{\Phi_l}}\sum_{n=1}^\infty\Big(\big(1-\mathcal{P}_{\rm sf}^{\Phi_l,\Phi_e}\big)\mathcal{P}_{\rm cf}^{\Phi_l}\Big)^n\Big] \nonumber\\
&=&\mathbb{E}_{\Phi_l,\Phi_e}\bigg[\frac{\mathcal{P}_{\rm sf}^{\Phi_l,\Phi_e}}{1-\mathcal{P}_{\rm cf}^{\Phi_l}+\mathcal{P}_{\rm sf}^{\Phi_l,\Phi_e}\mathcal{P}_{\rm cf}^{\Phi_l}} \bigg], \label{eqn:sop}
\end{eqnarray}
where $(a)$ holds since $N_r^{\Phi_l}$ is a geometric distributed random variable with success probability $1-\mathcal{P}_{\rm cf}$.
The last equation holds due to the sum formula of infinite geometric series whose successive terms have a common ratio.
With the above discussion, we could further derive the secrecy outage probability in the following theorem.
\begin{thm}
\label{thm:2}
The secrecy outage probability of the typical link in the backlogged scenario can be approximated by
\begin{equation}
\mathcal{P}_{\rm so} \simeq \frac{\exp\big(\frac{\lambda_e}{p\lambda_l \theta_e^\delta C(\delta)}\big)-1}{\exp\big(\frac{\lambda_e}{p\lambda_l \theta_e^\delta C(\delta)}\big)+\exp\big(-
p\lambda_l \pi r_0^2\theta_t^\delta C(\delta)\big)-1}. \label{eqn:securityoutage}
\end{equation}
\end{thm}
\begin{proof}
See Appendix \ref{appendix:a}.
\end{proof}

The secrecy outage probability given by Theorem \ref{thm:2} is closed-form, and the relationship between the secrecy outage probability and the system parameters can be observed directly from the equation (\ref{eqn:securityoutage}). The equation (\ref{eqn:securityoutage}) also verifies the intuition that the secrecy outage probability is positively correlated with the intensity of the eavesdroppers $\lambda_e$ and is negative correlated with the intensity of the legitimate transmitters $\lambda_l$. Specially, we get $\mathcal{P}_{\rm so}\rightarrow0$ as $\lambda_e\rightarrow0$, and $\mathcal{P}_{\rm so}\rightarrow1$ as $\lambda_e\rightarrow+\infty$. The reason is that increasing the intensity of eavesdroppers may increase the probability that an eavesdropper appears nearby the legitimate links. Meanwhile, increasing the intensity of legitimate transmitters may increase the interference in the wireless network, thus preventing the message from being eavesdropped.
Note that when either $p\rightarrow0$ or $\lambda_l\rightarrow0$ holds, the secrecy outage probability goes to one, i.e., $\mathcal{P}_{\rm so}\rightarrow1$, indicating that the secrecy outage event happens with high probability when there are very few active legitimate transmitters.
This observation illustrates that increasing the number of concurrent transmissions in the wireless networks is helpful to improve the secrecy performance, i.e., increasing the interference in the wireless networks is helpful from the aspect of the security performance.

\subsection{Dynamic Scenario}
In the dynamic scenario, the delay consists of the queueing delay and the transmission delay. The analysis of the delay in the dynamic scenario requires the combination of the stochastic geometry and the queueing theory, which induces the interacting queues problem.
In order to bypass these difficulties, we propose an approach to approximate the delay in the dynamic scenario. We first simplify the interaction between the queues and assume that all interfering transmitters are active with the same probability $q$. Then, we derive the connection failure probability $\mathcal{P}_{\rm cf}$ for the typical legitimate link in each time slot. The service rate of the queue at the typical legitimate transmitter will be $p(1-\mathcal{P}_{\rm cf})$, which means that the probability for a packet at the typical user being scheduled and successfully transmitted in a time slot is $p(1-\mathcal{P}_{\rm cf})$.
Then, $\rho=\frac{\xi}{p(1-\mathcal{P}_{\rm cf})}$ represents the average proportion of time occupied by a transmitter.
When the queue is empty, the transmitter is inactive, and when the transmitter is occupied, it is active with probability $p$ due to the random access.
Therefore, the active probability of the typical legitimate link is $\min\{p\rho,1\}=\min\{\frac{\xi}{1-\mathcal{P}_{\rm cf}},1\}$.
Note that $\mathcal{P}_{\rm cf}$ is a function of $q$. By solving the equation $\min\{\frac{\xi}{1-\mathcal{P}_{\rm cf}},1\}=q$, we obtain a solution $q=q^\star$ which could be used to derive the approximated delay for the wireless network.

Since the transmissions of the typical user in different time slots are affected by the independent fading and the independent random access, the probability for a transmission attempt to be scheduled and successful in all time slots is the same, which is $p(1-\mathcal{P}_{\rm cf})$.
Therefore, the queueing system at the typical transmitter is a Geo/G/1 queue, or a discrete-time single server retrial queue \cite{atencia2004discrete, falin1990survey}. In the Geo/G/1 queue, the arrival process of the packets is a Bernoulli process with intensity $\xi$ packets per time slot. The arrival process of the packets is called the geometric arrival since the probability that a packet arrives in a time slot is $\xi$, and the number of time slots between two adjacent arrivals is a geometric random variable. The success probability is $p(1-\mathcal{P}_{\rm cf})$ and the service times of packets are i.i.d. with geometric distribution. From \cite{atencia2004discrete}, we get the mean delay $\mathcal{D}$ for the Geo/G/1 queue with given $\mathcal{P}_{\rm cf}$ as
\begin{eqnarray}
\mathcal{D} = \left\{ \begin{array}{ll}
\frac{1-\xi}{p(1-\mathcal{P}_{\rm cf})-\xi} & \textrm{if $p(1-\mathcal{P}_{\rm cf})> \xi$}\\
\infty & \textrm{if $p(1-\mathcal{P}_{\rm cf})\leq \xi$}.
\end{array} \right. \label{eqn:randomdelay}
\end{eqnarray}

With the above discussions, we could obtain the mean delay in the dynamic scenario in the following theorem.
\begin{thm}
\label{thm:delay2}
The mean delay of the typical legitimate link in the dynamic scenario is
\begin{equation}
\mathcal{D} \simeq \frac{1-1/\xi}{1+\frac{p\lambda_l\pi r_0^2 \theta_t^\delta C(\delta)}{\mathcal{W}(-\xi\lambda_l\pi r_0^2 \theta_t^\delta C(\delta))}}. \label{eqn:dyndelay}
\end{equation}
\end{thm}
\begin{proof}
By assuming that all interfering transmitters are active with the same probability $q$, we get the connection failure probability as
\begin{eqnarray}
\mathcal{P}_{\rm cf}&=&1-\mathbb{E}_{\Phi_l}\Big[\prod_{y\in\Phi_l}\Big(\frac{q}{1+\theta_t r_0^\alpha|y|^{-\alpha}}+1-q\Big)\Big]\nonumber\\
&=&1-\exp\big(-q\lambda_l\pi r_0^2 \theta_t^\delta C(\delta) \big).
\end{eqnarray}
Considering the equation $\min\{\frac{\xi}{1-\mathcal{P}_{\rm cf}},1\}=q$, we get
\begin{eqnarray}
\min\bigg\{\frac{\xi}{\exp\big(-q\lambda_l\pi r_0^2 \theta_t^\delta C(\delta) \big)},1\bigg\}=q.
\end{eqnarray}

Solving the above equation, we get the solution
\begin{eqnarray}
q^\star=\min\bigg\{-\frac{\mathcal{W}(-\xi \lambda_l\pi r_0^2 \theta_t^\delta C(\delta))}{\lambda_l\pi r_0^2 \theta_t^\delta C(\delta)},1\bigg\}.
\end{eqnarray}
where $\mathcal{W}(z)$ is the Lambert-$\mathcal{W}$ function with the defining equation $\mathcal{W}(z)\exp(\mathcal{W}(z))=z$.
Then, the connection failure probability is
\begin{eqnarray}
\mathcal{P}_{\rm cf}=1-\frac{\xi}{q^\star}=\min\bigg\{1+\frac{\xi\lambda_l\pi r_0^2 \theta_t^\delta C(\delta)}{\mathcal{W}(-\xi\lambda_l\pi r_0^2 \theta_t^\delta C(\delta))},1-{\xi}\bigg\}.
\end{eqnarray}
Having derived the connection failure probability, the mean delay is
\begin{eqnarray}
\mathcal{D} =
\frac{1-\xi}{p(1-\mathcal{P}_{\rm cf})-\xi} =\frac{1-1/\xi}{1+\frac{p\lambda_l\pi r_0^2 \theta_t^\delta C(\delta)}{\mathcal{W}(-\xi\lambda_l\pi r_0^2 \theta_t^\delta C(\delta))}}. \label{eqn:dyndelay2}
\end{eqnarray}
Therefore, we get the results in the theorem.
\end{proof}

It is hard to interpret the relationship between the mean delay and the system parameters since the property of the Lambert-$\mathcal{W}$ function is not intuitive. However, we could further simplify the above results for some special cases.
Note that when $x\rightarrow0$, the asymptotic expansions of the Lambert-$\mathcal{W}$ function is
\begin{eqnarray}
\mathcal{W}(x)=\sum_{n=1}^\infty\frac{(-n)^{n-1}}{n!}=x-x^2+\frac{3}{2}x^3-\frac{8}{3}x^4+... . \label{eqn:expand}
\end{eqnarray}
Therefore, we have $\mathcal{W}(x)\sim x$ when $x\rightarrow0$. When $\lambda_l r_0^2 \theta_t^\delta C(\delta)\rightarrow0$, the result given by Theorem \ref{thm:delay2} can be simplified into
\begin{eqnarray}
\mathcal{D} \simeq \frac{1-1/\xi}{1+\frac{p\lambda_l\pi r_0^2 \theta_t^\delta C(\delta)}{\mathcal{W}(-\xi\lambda_l\pi r_0^2 \theta_t^\delta C(\delta))}} \simeq \frac{1-1/\xi}{1-p/\xi}
\end{eqnarray}

The above result indicates that when either $\lambda_l\rightarrow0$, $r_0\rightarrow0$, or $\theta_t\rightarrow0$ holds, the delay could be approximated into simple form. Then, the relationship between the mean delay $\mathcal{D}$, the transmit probability $p$ and the arrival rate $\xi$ could be easily obtained.

The secrecy failure event happens in the transmitting process and is not affected by the queueing process.
Thus, the secrecy outage probability in the dynamic scenario is similar to that in the backlogged scenario since the secrecy outage is only determined by the delivering process of a packet.
The only difference between the secrecy outage probability in the backlogged scenario and that in the dynamic scenario is that the active probabilities for the two scenarios are different.
In the backlogged scenario, the active probability for the interfering transmitters is the transmit probability $p$ since all transmitters are backlogged and always have packets to transmit. However, in the dynamic scenario, the active probability is $q^\star$. Therefore, we get the following theorem, which gives the secrecy outage probability in the dynamic scenario.

\begin{thm}
\label{thm:4}
The secrecy outage probability of the typical link in the dynamic scenario can be approximated by
\begin{equation}
\mathcal{P}_{\rm so} \simeq \frac{\exp\big(\frac{\lambda_e}{q^\star\lambda_l \theta_e^\delta C(\delta)}\big)-1}{\exp\big(\frac{\lambda_e}{q^\star\lambda_l \theta_e^\delta C(\delta)}\big)+\exp\big(-
q^\star\lambda_l \pi r_0^2\theta_t^\delta C(\delta)\big)-1}, \label{eqn:th4pso}
\end{equation}
where $q^\star$ is given by
\begin{eqnarray}
q^\star=\min\bigg\{-\frac{\mathcal{W}(-\xi \lambda_l\pi r_0^2 \theta_t^\delta C(\delta))}{\lambda_l\pi r_0^2 \theta_t^\delta C(\delta)},1\bigg\}. \label{eqn:qstar}
\end{eqnarray}
\end{thm}
\begin{proof}
The proof is similar to that of Theorem \ref{thm:2}.
\end{proof}

The equation (\ref{eqn:th4pso}) shows that $\mathcal{P}_{\rm so} \simeq 0$ as $\lambda_e\rightarrow0$ or $\theta_e\rightarrow\infty$, indicating that when the number of eavesdroppers is very small or when the SIR threshold to eavesdrop the message is very large, the secrecy outage event hardly ever happens. On the other hand, when $\lambda_e\rightarrow\infty$ or $\theta_e\rightarrow0$, we get $\mathcal{P}_{\rm so} \simeq 1$, indicating that when the eavesdroppers are densely deployed or when the messages are easily intercepted by the eavesdroppers, the secrecy outage event almost always happens.

Similar to the discussions after Theorem \ref{thm:delay2}, we can also simplify the result given by Theorem \ref{thm:4}.
Note that $\mathcal{W}(x)\sim x$ as $x\rightarrow0$. When $r_0^2 \theta_t^\delta C(\delta)\rightarrow0$, the active probability given by the equation (\ref{eqn:qstar}) can be simplified as
\begin{eqnarray}
q^\star=\min\{{\xi},1\}.
\end{eqnarray}
Note that if $\xi>1$, the arrival rate is larger than the transmit probability, thus, the queues in the network will be unstable.
A more relevant case is $\xi<1$, which is a necessary condition for the queues to be stable. If both $ r_0^2 \theta_t^\delta C(\delta)\rightarrow0$ and $\xi<1$ hold, the active probability will be
\begin{eqnarray}
q^\star={\xi}.
\end{eqnarray}
Plugging $q^\star={\xi}$ into (\ref{eqn:th4pso}), we get the secrecy outage probability when $r_0^2 \theta_t^\delta C(\delta)\rightarrow0$ as
\begin{eqnarray}
\mathcal{P}_{\rm so} &\simeq& \frac{\exp\big(\frac{\lambda_e}{{\xi}\lambda_l \theta_e^\delta C(\delta)}\big)-1}{\exp\big(\frac{\lambda_e}{{\xi}\lambda_l \theta_e^\delta C(\delta)}\big)+\exp\big(-
{\xi}\lambda_l \pi r_0^2\theta_t^\delta C(\delta)\big)-1} \nonumber\\
&\simeq& \frac{\exp\big(\frac{\lambda_e}{{\xi}\lambda_l \theta_e^\delta C(\delta)}\big)-1}{\exp\big(\frac{\lambda_e}{{\xi}\lambda_l \theta_e^\delta C(\delta)}\big)}\nonumber\\
&=& 1- \exp\Big(-\frac{\lambda_e}{{\xi}\lambda_l \theta_e^\delta C(\delta)}\Big). \label{eqn:th4pso2}
\end{eqnarray}
The relationship between the secrecy outage probability and the system parameters, such as the intensity of eavesdroppers $\lambda_e$, the intensity of legitimate transmitters $\lambda_l$, the SIR threshold $\theta_e$ and so on, could be observed from the equation (\ref{eqn:th4pso2}) directly.

\section{Performance of Message Split}
\label{sec:perform}
In this section, we propose and analyze a simple transmission mechanism, in which a confidential message is divided into two packets to be delivered independently. A confidential message is successfully decoded at the legitimate receiver only when both of the two packets from the message are successfully decoded at the receiver. Meanwhile, the perfect security could be guaranteed if at least one of the two packets from the same message cannot be decoded by an eavesdropper.
This mechanism has a direct effect on the delay and may also affect the security performance. Therefore, in this section, we explore this simple transmission mechanism to provide a guidance to tradeoff the delay and the security of the wireless network.

\subsection{Backlogged Scenario}
Since one message is split into two packets, the size of each packet is halved. When fixing the duration of the time slots, the requirement for the rate of the confidential message could be reduced from $R_s$ to $R_s'=R_s/2$. The threshold for the rate of the eavesdropping link is also reduced from $R_e$ to $R_e'=R_e/2$. Then, the rate threshold of codewords becomes $R_t'=R_s'+R_e'=(R_s+R_e)/2$.

Note that the mean delay is doubled since both of the two packets should be successful. By the similar derivations as Theorem \ref{thm:1}, we get the following corollary which gives the mean delay in the case of message split.

\begin{cor}
\label{cor:1}
The mean delay of the typical legitimate link in the backlogged scenario with message split is
\begin{equation}
\mathcal{D} = \frac{2}{p}\exp\left({\lambda_l \pi r_0^2\theta_t'^{\delta}C(\delta)p}{(1-p)^{\delta-1}}\right), \label{eqn:associateP1}
\end{equation}
where $\theta_t'=2^{R_t'}-1$.
\end{cor}

Comparing with the result given by (\ref{eqn:localdelay}), we observe that the mean delay is doubled and $\theta_t$ is replaced by $\theta_t'$.
However, since the effect of the two changes are in the opposite direction, whether the mean delay is increased or not is still uncertain, which depends on the parameters such as the transmit probability $p$, the intensity of legitimate transmitters $\lambda_l$, and so on.

In order to intercept the confidential message, both of the two packets from the same message should be successfully decoded at the eavesdropper. Therefore, by the similar derivations as that in Theorem \ref{thm:2}, we get
\begin{cor}
\label{cor:2}
The secrecy outage probability of the typical link in the backlogged scenario with message split can be approximated by
\begin{equation}
\mathcal{P}_{\rm so} \simeq \bigg(\frac{\exp\big(\frac{\lambda_e}{p\lambda_l \theta_e'^\delta C(\delta)}\big)-1}{\exp\big(\frac{\lambda_e}{p\lambda_l \theta_e'^\delta C(\delta)}\big)+\exp\big(-
p\lambda_l \pi r_0^2\theta_t'^\delta C(\delta)\big)-1}\bigg)^2,
\end{equation}
where $\theta_e'=2^{R_e'}-1$.
\end{cor}

The secrecy outage probability in the backlogged scenario with message split is closed-form. The result also reveals that $\mathcal{P}_{\rm so}\rightarrow0$ as $\lambda_e\rightarrow0$, and $\mathcal{P}_{\rm so}\rightarrow1$ as $\lambda_e\rightarrow+\infty$. Though whether the message split decreases the secrecy outage probability is still not intuitive from Corollary \ref{cor:2}, the numerical analysis could be conducted to evaluate the effect of message split.

\subsection{Dynamic Scenario}
The analysis of the dynamic scenario is different from that of the backlogged scenario when the message split is introduced in the transmission mechanism.
The arrival process of the messages is still a Bernoulli process of arrival rate $\xi$. Since each message is split into two packets which are transmitted independently, at least two time slots are needed to successfully deliver one message. However, the queueing system is still a Geo/G/1 queue because the number of
time slots between two adjacent arrivals is still a geometric random variable.
When each message is divided into two packets, the service time of each message is the summation of the service time of the two packets. Since the service time of each packet is a geometric random variable,
the total service time for each message is the summation of two geometric distributed random variables.
The analysis of this queueing system is rather difficult because the service rate is not a well know distribution. Thus, we use the formula for the mean delay obtained from the M/M/1 queueing system to approximate the mean delay in the dynamic scenario with message split. With these discussions, we get the following corollary, which gives the mean delay in the dynamic scenario with message split.

\begin{cor}
\label{cor:delay2}
The mean delay of the typical legitimate link in the dynamic scenario with message split is
\begin{equation}
\mathcal{D} \simeq  -\frac{1}{\frac{p\xi\lambda_l\pi r_0^2 \theta_t'^\delta C(\delta)}{2\mathcal{W}(-\xi\lambda_l\pi r_0^2 \theta_t'^\delta C(\delta))}+\xi}, \label{eqn:dyndelay1}
\end{equation}
where $\theta_t'=2^{R_t'}-1$.
\end{cor}
\begin{proof}
By approximating the queueing system with the M/M/1 model, the mean delay can be approximated as
\begin{equation}
\mathcal{D} \simeq \frac{1}{\mu-\xi}, \label{eqn:mm1}
\end{equation}
where $\xi$ is the arrival rate of the messages, and $\mu$ is the service rate.
The mean delay in the dynamic scenario with message split could be obtained by the service rate of the Geo/G/1 queue.

Similar to the derivations in Theorem \ref{thm:delay2}, the connection failure probability of a packet is
\begin{eqnarray}
\mathcal{P}_{\rm cf}&=&\min\bigg\{1+\frac{\xi\lambda_l\pi r_0^2 \theta_t'^\delta C(\delta)}{\mathcal{W}(-\xi\lambda_l\pi r_0^2 \theta_t'^\delta C(\delta))},1-{\xi}\bigg\},
\end{eqnarray}
where $\theta_t'=2^{R_t'}-1$. The service time for each packet is a geometric random variable with success probability $p(1-\mathcal{P}_{\rm cf})$, and the mean service time for each packet is $\frac{1}{p(1-\mathcal{P}_{\rm cf})}$.
Thus, the mean service time for each message is $\frac{2}{p(1-\mathcal{P}_{\rm cf})}$, and the mean service rate for each message is $\frac{p(1-\mathcal{P}_{\rm cf})}{2}$.
Plugging into the equation (\ref{eqn:mm1}), we get the mean delay as
\begin{eqnarray}
\mathcal{D} &\simeq&  -\frac{1}{\frac{p\xi\lambda_l\pi r_0^2 \theta_t'^\delta C(\delta)}{2\mathcal{W}(-\xi\lambda_l\pi r_0^2 \theta_t'^\delta C(\delta))}+\xi}.
\end{eqnarray}
Therefore, we get the result in the corollary.
\end{proof}

By the expansion of the Lambert-$\mathcal{W}$ function given by the equation (\ref{eqn:expand}), we get the mean delay when $\lambda_l\pi r_0^2 \theta_t'^\delta C(\delta)\rightarrow0$ as
\begin{eqnarray}
\mathcal{D} \simeq \frac{1}{\frac{p}{2(1+\xi\lambda_l\pi r_0^2 \theta_t'^\delta C(\delta))}-\xi}\simeq \frac{2}{p-2\xi}.
\end{eqnarray}

Furthermore, when $p=1$ and $\xi\rightarrow0$, the mean delay $\mathcal{D}$ is exactly two time slots, which is the case where the transmission attempts of a packet is scheduled with probability one and is always successful when it is scheduled.

To eavesdrop the confidential message at an eavesdropper in the dynamic scenario with message split, both of the two packets from the same message should be successfully decoded at the eavesdropper. By the similar derivations as that in Theorem \ref{thm:4}, we get
\begin{cor}
\label{cor:4}
The secrecy outage probability of the typical link in the dynamic scenario with message split can be approximated by
\begin{equation}
\mathcal{P}_{\rm so} \simeq \bigg(\frac{\exp\big(\frac{\lambda_e}{q^\star\lambda_l \theta_e'^\delta C(\delta)}\big)-1}{\exp\big(\frac{\lambda_e}{q^\star\lambda_l \theta_e'^\delta C(\delta)}\big)+\exp\big(-
q^\star\lambda_l \pi r_0^2\theta_t'^\delta C(\delta)\big)-1}\bigg)^2,
\end{equation}
where $q^\star$ is given by
\begin{eqnarray}
q^\star=\min\bigg\{-\frac{\mathcal{W}(-\xi \lambda_l\pi r_0^2 \theta_t'^\delta C(\delta))}{\lambda_l\pi r_0^2 \theta_t'^\delta C(\delta)},1\bigg\},
\end{eqnarray}
$\theta_t'=2^{R_t'}-1$, and $\theta_e'=2^{R_e'}-1$.
\end{cor}

Specially, when $\lambda_l\pi r_0^2 \theta_t'^\delta C(\delta)\rightarrow0$, using the expansion of the Lambert-$\mathcal{W}$ function given by (\ref{eqn:expand}), we obtain
\begin{eqnarray}
q^\star=\min\{{\xi},1\},
\end{eqnarray}
When the queue is stable, the arrival rate of messages must be smaller than one, i.e., $\xi<1$. Therefore, we have
\begin{eqnarray}
q^\star={\xi}.
\end{eqnarray}
Then, the secrecy outage probability given by Corollary \ref{cor:4} could be simplified as
\begin{eqnarray}
\mathcal{P}_{\rm so} &\simeq& \bigg(\frac{\exp\big(\frac{\lambda_e}{\xi\lambda_l \theta_e'^\delta C(\delta)}\big)-1}{\exp\big(\frac{\lambda_e}{\xi\lambda_l \theta_e'^\delta C(\delta)}\big)+\exp\big(-
{\xi}\lambda_l \pi r_0^2\theta_t'^\delta C(\delta)\big)-1}\bigg)^2\nonumber\\
&\simeq& \bigg({1-\exp\bigg(-\frac{\lambda_e}{\xi\lambda_l \theta_e'^\delta C(\delta)}\bigg)}\bigg)^2.
\end{eqnarray}

\section{Numerical Evaluation}
\label{sec:numerical}
In this section, we evaluate the tradeoff between the mean delay and the secrecy outage probability through numerical analysis to gain insight.
Without special explanation, the numerical parameters are set as follows in default. The intensities of the legitimate transmitters and the eavesdroppers are set as $\lambda_l=0.05$ and $\lambda_e=0.01$, respectively, and the desired link distance is $r_0=1$. The path loss exponent is $\alpha=4$, the transmit probability is $p=0.8$, and the packet arrival rate is $\xi=0.1$.
The rate threshold of codewords is $R_t=3$ bits/Hz, and that of confidential messages is $R_e=1$ bit/Hz.

\begin{figure}
\centering
\includegraphics[width=0.75\textwidth]{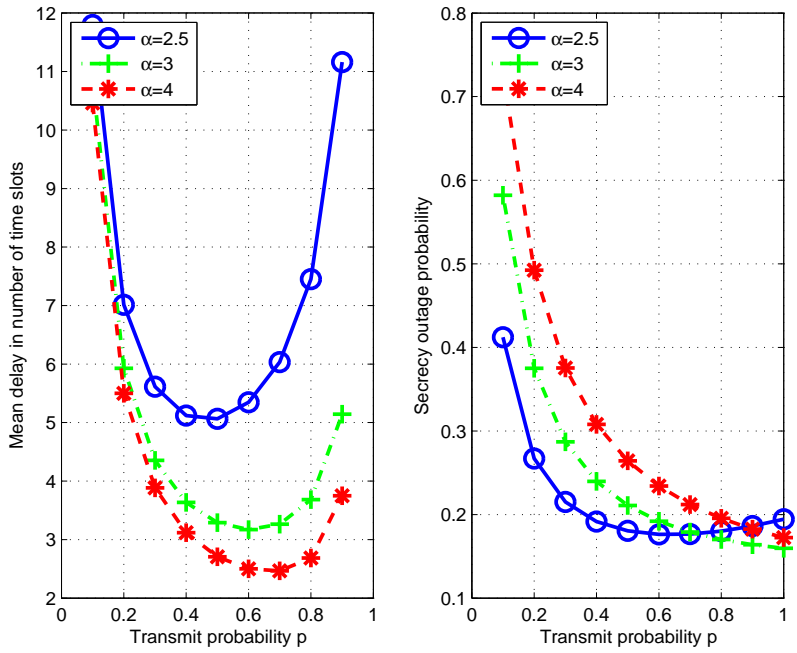}
\caption{The mean delay and the secrecy outage probability as functions of the transmit probability $p$ in the backlogged scenario.}
\label{fig:Bl_p}
\end{figure}

\begin{figure}
\centering
\includegraphics[width=0.75\textwidth]{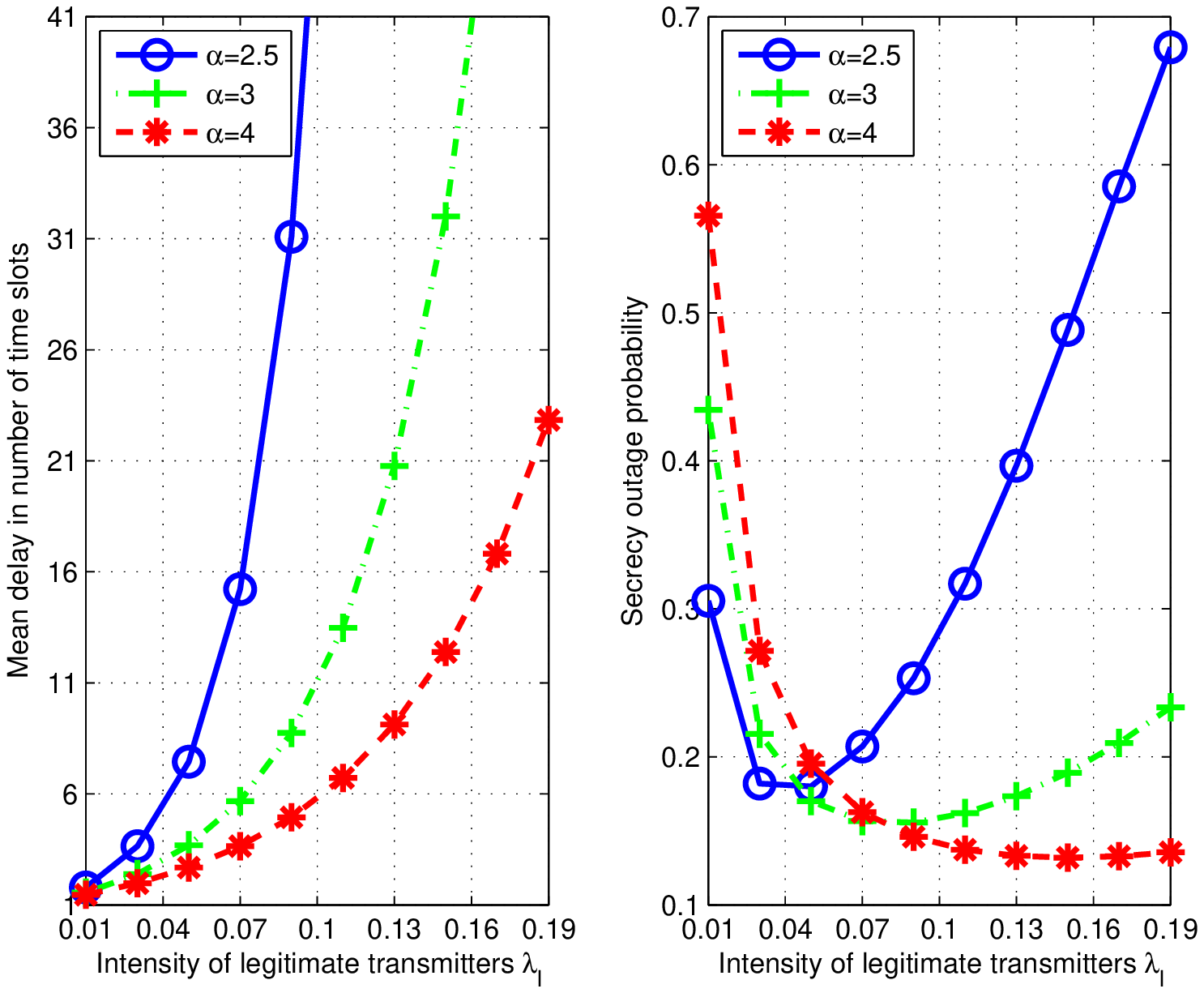}
\caption{The mean delay and the secrecy outage probability as functions of the intensity of legitimate transmitters $\lambda_l$ in the backlogged scenario.}
\label{fig:Bl_lambdal}
\end{figure}

\begin{figure}
\centering
\includegraphics[width=0.75\textwidth]{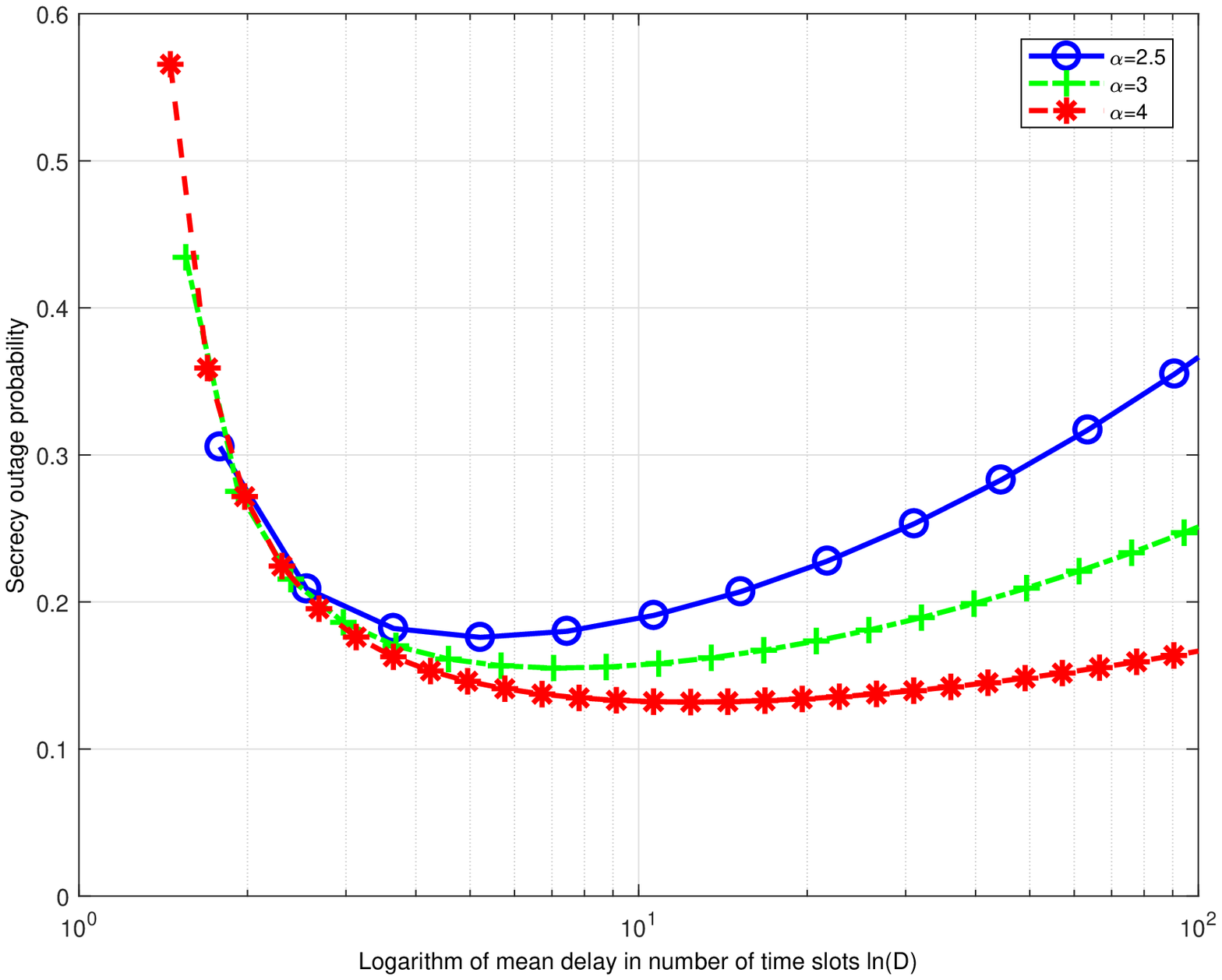}
\caption{Tradeoff between the mean delay and the secrecy outage probability when $\lambda_l$ increases from $0.01$ to $0.09$ in the backlogged scenario.}
\label{fig:Bl_tradeoff}
\end{figure}

Figure \ref{fig:Bl_p} shows the mean delay and the secrecy outage probability as functions of the transmit probability $p$ in the backlogged scenario. As $p$ increases, both the mean delay and the secrecy outage probability first decrease then increase. When $p$ starts to increase, the mean delay decreases since the successful delivery of a packet is limited by the transmit probability, while the secrecy outage probability decreases since the interference for the eavesdropping links increases. When $p$ continues to grow, the mean delay increases since the interference increases, while the secrecy outage probability increases since the probability to eavesdrop the message is increased with large $p$.
Figure \ref{fig:Bl_lambdal} shows the mean delay and the secrecy outage probability as functions of $\lambda_l$ in the backlogged scenario. When $\lambda_l$ increases, the mean delay increases and the secrecy outage probability decreases since the interference in the network increases. However, when $\lambda_l$ continues to grow, the secrecy outage probability increases since the distances of the eavesdropping links are decreased.
Figure \ref{fig:Bl_tradeoff} plots the tradeoff between the mean delay and the secrecy outage probability when increasing $\lambda_l$ in the backlogged scenario.
Though the value of $\lambda_l$ is not marked in Fig. \ref{fig:Bl_tradeoff}, we can observe from Theorem \ref{thm:1} that $\lambda_l$ is linear positive correlated to ln$(D)$ which is the value of the horizontal ordinate.
Note that when the path loss exponent $\alpha$ increases, the secrecy outage probability decreases for large $\lambda_l$ and increases for small $\lambda_l$, indicating that the security performance is much better when the signal attenuates quickly (i.e., when $\alpha$ is large) for large $\lambda_l$, and it is reverse for small $\lambda_l$.

%\begin{figure}
%\centering
%\includegraphics[width=0.85\textwidth]{Dynamic_xi.eps}
%\caption{.}
%\label{fig:Dynamic_xi}
%\end{figure}

\begin{figure}
\centering
\includegraphics[width=0.75\textwidth]{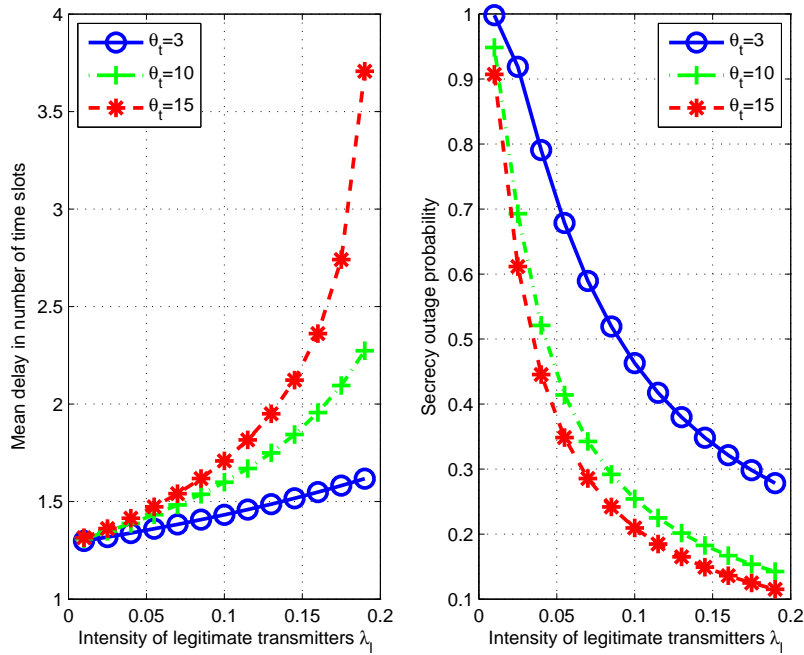}
\caption{The mean delay and the secrecy outage probability as functions of the intensity of legitimate transmitters $\lambda_l$ in the dynamic scenario.}
\label{fig:Dynamic_theta_t}
\end{figure}

Figure \ref{fig:Dynamic_theta_t} shows the mean delay and the secrecy outage probability as functions of the intensity of legitimate transmitters $\lambda_l$ with different $\theta_t$ in the dynamic scenario. When varying $\theta_t$, we keep the rate threshold of confidential messages $R_s$ to be the same (i.e., $R_s=1$ bits/Hz). Thus, the minimal confidential rate $R_s=1$ is guaranteed. Figure \ref{fig:Dynamic_theta_t} reveals that although the same confidential rate is guaranteed, the mean delay is larger and the secrecy outage probability is smaller in the high SIR regime than those in the low SIR regime.

\begin{figure}
\centering
\includegraphics[width=0.75\textwidth]{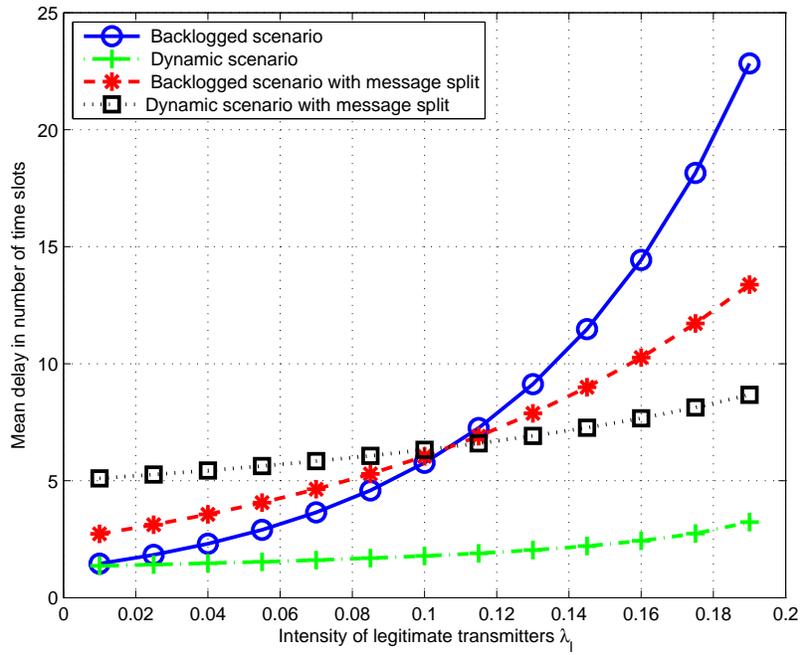}
\caption{Comparison of the mean delay in different scenarios with and without message split.}
\label{fig:Compare_split_delay}
\end{figure}

\begin{figure}
\centering
\includegraphics[width=0.75\textwidth]{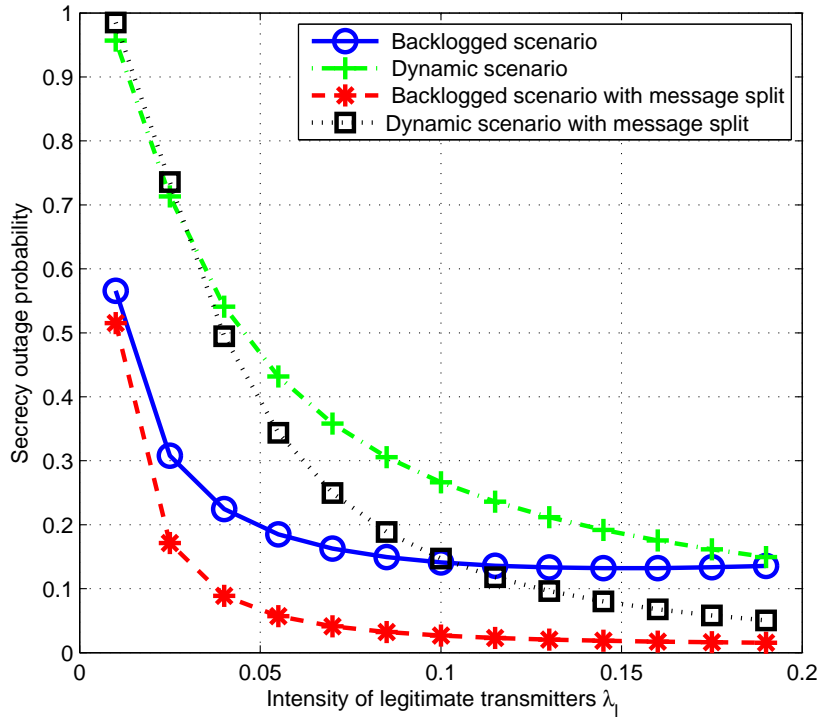}
\caption{Comparison of the secrecy outage probability in different scenarios with and without message split.}
\label{fig:Compare_split_so}
\end{figure}

Figure \ref{fig:Compare_split_delay} and Figure \ref{fig:Compare_split_so} show the comparison of the mean delay and the secrecy outage probability for different scenarios with and without message split. The arrival rate is $\xi=0.2$. Figure \ref{fig:Compare_split_delay} reveals that in the backlogged scenario, the mean delay with message split is larger than that without message split when $\lambda_l$ is small, while it is reverse for large $\lambda_l$, indicating that in the backlogged scenario, message split is beneficial for the delay performance only for large $\lambda_l$. Figure \ref{fig:Compare_split_delay} also reveals that in the dynamic scenario, the delay performance without message split is always better than that with message split.
Figure \ref{fig:Compare_split_so} shows that in the backlogged scenario, introducing message split greatly decreases the secrecy outage probability and improves the security performance, while in the dynamic scenario, introducing message split is only beneficial when $\lambda_l$ is large. Accordingly, in the actual application scenario such as the military communications, the message split scheme could be introduced to improve the security performance at the expense of increasing the delay when the density of the soldiers in the communications is not very small.

\section{Conclusions}
\label{sec:conclusions}
In this work, we evaluate the tradeoff between the delay and the security performance in large wireless networks for both the backlogged scenario and the dynamic scenario.
By combining the tools from stochastic geometry and queueing theory, we derive the close-formed results for the mean delay and the secrecy outage probability. We also analyzed the effect of a simple transmission mechanism which splits a message into two packets.

Our results reveal that the security performance of the wireless network is better when the path loss exponent is larger. We also observe that under the condition that a certain confidential rate is guaranteed, the delay performance is better in high SIR regime than that in low SIR regime.
Moreover, by introducing the simple mechanism of message split, it is shown that the delay performance in the backlogged scenario is improved when the density of the legitimate transmitters is large.
Meanwhile, introducing message split greatly improves the
security performance in the backlogged scenario, while slightly improves the security performance in the dynamic scenario when the density of the legitimate transmitters is large.

\begin{appendices}
\section{Proof of Theorem \ref{thm:2}}
\label{appendix:a}
\end{appendices}
Let $\Phi_k\subset\Phi_l$ be the set of transmitters that are active in time slot $k$, and let $I_k$ be the interference at the typical legitimate receiver located at $x_0$ in time slot $k$. Then, the mathematical expression for $I_k$ is
\begin{equation}
I_k=\sum_{y\in\Phi_l}h_y|y-x_0|^{-\alpha}\mathbf{1}(y\in\Phi_k).
\end{equation}
With given $\Phi_l$, the connection failure probability is
\begin{eqnarray}
\mathcal{P}_{\rm cf}^{\Phi_l}&{=}&\mathbb{P}\big(h_{x_0}r_0^{-\alpha}<\theta_t I_k\mid\Phi_l\big) \nonumber\\
&\stackrel{(a)}{=}&1-\mathbb{E}\big[\exp\big(-\theta_t r_0^\alpha I_k\big)\mid\Phi_l\big] \nonumber\\
&=&1-\mathbb{E}\bigg[\exp\bigg(-\sum_{y\in\Phi_l}\theta_t r_0^\alpha h_y|y-x_0|^{-\alpha}\mathbf{1}(y\in\Phi_k)\bigg)\mid\Phi_l\bigg] \nonumber\\
&=&1-\prod_{y\in\Phi_l}\Big(p\mathbb{E}\Big[\exp\big(-\theta_t r_0^\alpha h_y|y-x_0|^{-\alpha}\big)\mid\Phi_l\Big]+1-p\Big) \nonumber\\
&\stackrel{(b)}{=}&1-\prod_{y\in\Phi_l}\Big(\frac{p}{1+\theta_t r_0^\alpha|y-x_0|^{-\alpha}}+1-p\Big)\nonumber\\
&\stackrel{(c)}{=}&1-\prod_{y\in\Phi_l}\Big(\frac{p}{1+\theta_t r_0^\alpha|y|^{-\alpha}}+1-p\Big),
\label{equ:succ_aloha}
\end{eqnarray}
where $(a)$ and $(b)$ follow because the fading coefficients $\{h_y\}$ are i.i.d. exponential distributed random variables with unit mean, $(c)$ holds due to the stationarity of the PPP.

With given $\Phi_l$ and $\Phi_e$, the secrecy failure probability is
\begin{eqnarray}
\mathcal{P}_{\rm sf}^{\Phi_l,\Phi_e}&=& 1-\prod_{x_e\in\Phi_e}\mathbb{P}\{\mathrm{SIR}_{x_e}<\theta_e\mid \Phi_l,\Phi_e\} \nonumber\\
&=& 1-\prod_{x_e\in\Phi_e}\mathbb{P}\big\{h_e|x_e|^{-\alpha}<\theta_eI_k\mid \Phi_l,\Phi_e\big\} \nonumber\\
&=& 1-\prod_{x_e\in\Phi_e}\Big(1-\mathbb{E}\big[\exp(-\theta_eI_k|x_e|^{\alpha})\mid \Phi_l,\Phi_e\big]\Big) \nonumber\\
&=& 1-\prod_{x_e\in\Phi_e}\Big(1-\mathbb{E}\big[\exp(-\theta_e|x_e|^{\alpha}\sum_{y\in\Phi_l}h_y|y-x_e|^{-\alpha}\mathbf{1}(y\in\Phi_k))\mid \Phi_l,\Phi_e\big]\Big) \nonumber\\
&=& 1-\prod_{x_e\in\Phi_e}\Big(1-\prod_{y\in\Phi_l}\mathbb{E}\big[\exp(-\theta_e|x_e|^{\alpha}h_y|y-x_e|^{-\alpha}\mathbf{1}(y\in\Phi_k))\mid \Phi_l,\Phi_e\big]\Big) \nonumber\\
&=& 1-\prod_{x_e\in\Phi_e}\Big(1-\prod_{y\in\Phi_l}\Big(\frac{p}{1+\theta_e|x_e|^{\alpha}|y-x_e|^{-\alpha}}+1-p\Big)\Big).
\end{eqnarray}
From equation (\ref{eqn:sop}), we get the secrecy outage probability of the typical link as
\begin{eqnarray}
\mathcal{P}_{\rm so} &=& \mathbb{E}_{\Phi_l,\Phi_e}\bigg[\frac{\mathcal{P}_{\rm sf}^{\Phi_l,\Phi_e}}{1-\mathcal{P}_{\rm cf}^{\Phi_l}\big(1-\mathcal{P}_{\rm sf}^{\Phi_l,\Phi_e}\big)} \bigg]\nonumber \\
&=& \mathbb{E}_{\Phi_l,\Phi_e}\bigg[\bigg({1-\prod_{x_e\in\Phi_e}\Big(1-\prod_{y\in\Phi_l}\big(\frac{p}{1+\theta_e|x_e|^{\alpha}|y-x_e|^{-\alpha}}+1-p\big)\Big)}\bigg)\nonumber\\
&&\times\bigg(1-\bigg(1-\prod_{y\in\Phi_l}\big(\frac{p}{1+\theta_t r_0^\alpha|y|^{-\alpha}}+1-p\big)\bigg)\nonumber\\
&&\times\prod_{x_e\in\Phi_e}\Big(1-\prod_{y\in\Phi_l}\big(\frac{p}{1+\theta_e|x_e|^{\alpha}|y-x_e|^{-\alpha}}+1-p\big)\Big)\bigg)^{-1} \bigg]. \label{eqn:sop2}
\end{eqnarray}
To derive the exact secrecy outage probability is difficult, and we propose to approximate the secrecy outage probability.
Using the Jensen's inequality to approximate the secrecy outage probability, i.e., $\mathbb{E}_X(f(X))\simeq f_X(\mathbb{E}[X])$ for a random variable $X$ and a function $f(\cdot)$, we get
\begin{multline}
\mathcal{P}_{\rm so} \simeq \bigg({1-\mathbb{E}_{\Phi_l,\Phi_e}\bigg[\prod_{x_e\in\Phi_e}\Big(1-\prod_{y\in\Phi_l}\big(\frac{p}{1+\theta_e|x_e|^{\alpha}|y-x_e|^{-\alpha}}+1-p\big)\Big)\bigg]}\bigg)\\
\times\bigg(1-\bigg(1-\mathbb{E}_{\Phi_l}\bigg[\prod_{y\in\Phi_l}\big(\frac{p}{1+\theta_t r_0^\alpha|y|^{-\alpha}}+1-p\big)\bigg]\bigg)\\
\times\mathbb{E}_{\Phi_l,\Phi_e}\bigg[\prod_{x_e\in\Phi_e}\Big(1-\prod_{y\in\Phi_l}\big(\frac{p}{1+\theta_e|x_e|^{\alpha}|y-x_e|^{-\alpha}}+1-p\big)\Big)\bigg]\bigg)^{-1}.
\label{eqn:sop3}
\end{multline}

Using the probability generating functional (PGFL) of the PPP, we have
\begin{eqnarray}
\mathbb{E}_{\Phi_l}\bigg[\prod_{y\in\Phi_l}\big(\frac{p}{1+\theta_t r_0^\alpha|y|^{-\alpha}}+1-p\big)\bigg]&=&\exp\bigg(-\lambda_l\int_{\mathbb{R}^2}\bigg(1-\bigg({\frac{p}{1+\theta_t r_0^\alpha|y|^{-\alpha}}+1-p}\bigg)\bigg)\mathrm{d}y\bigg)\nonumber\\
&=&\exp\big(-p\lambda_l\pi r_0^2 \theta_t^\delta C(\delta) \big).
\end{eqnarray}
Due to the stationarity of the PPP and the independence between $\Phi_l$ and $\Phi_e$, we have
\begin{eqnarray}
&&\mathbb{E}_{\Phi_l,\Phi_e}\bigg[\prod_{x_e\in\Phi_e}\Big(1-\prod_{y\in\Phi_l}\big(\frac{p}{1+\theta_e|x_e|^{\alpha}|y-x_e|^{-\alpha}}+1-p\big)\Big)\bigg]\nonumber\\
&=&\mathbb{E}_{\Phi_e}\bigg[\prod_{x_e\in\Phi_e}\bigg(1-\mathbb{E}_{\Phi_l}\bigg[\prod_{y\in\Phi_l}\big(\frac{p}{1+\theta_e|x_e|^{\alpha}|y|^{-\alpha}}+1-p\big)\bigg]\bigg)\bigg]\nonumber\\
&=&\mathbb{E}_{\Phi_e}\bigg[\prod_{x_e\in\Phi_e}\Big(1-\exp\big(-p\lambda_l\pi |x_e|^2 \theta_e^\delta C(\delta) \big)\Big)\bigg]\nonumber\\
&=&\exp\Big(-\lambda_e\int_{\mathbb{R}^2}\exp\big(-p\lambda_l\pi |x|^2 \theta_e^\delta C(\delta) \big)\mathrm{d}x\Big)\nonumber\\
&=&\exp\Big(-\frac{\lambda_e}{p\lambda_l \theta_e^\delta C(\delta)}\Big).
\end{eqnarray}

Plugging into the equation (\ref{eqn:sop3}), we have
\begin{eqnarray}
\mathcal{P}_{\rm so} &\simeq& \frac{1-\exp\big(-\frac{\lambda_e}{p\lambda_l \theta_e^\delta C(\delta)}\big)}{1- \exp\big(
-\frac{\lambda_e}{p\lambda_l \theta_e^\delta C(\delta)}\big) +\exp\big(-
p\lambda_l \pi r_0^2\theta_t^\delta C(\delta)
-\frac{\lambda_e}{p\lambda_l \theta_e^\delta C(\delta)}\big)}\nonumber\\
&=& \frac{\exp\big(\frac{\lambda_e}{p\lambda_l \theta_e^\delta C(\delta)}\big)-1}{\exp\big(\frac{\lambda_e}{p\lambda_l \theta_e^\delta C(\delta)}\big)+\exp\big(-
p\lambda_l \pi r_0^2\theta_t^\delta C(\delta)\big)-1}.\label{eqn:sop4}
\end{eqnarray}
%In order to derive the upper bound for the security outrage probability, from equation (\ref{eqn:sop}), we have
%\begin{eqnarray}
%\mathcal{P}_{\rm so}&=&\mathbb{E}_{\Phi_l,\Phi_e}\bigg[\frac{\mathcal{P}_{\rm sf}^{\Phi_l,\Phi_e}}{1-\mathcal{P}_{\rm cf}^{\Phi_l}+\mathcal{P}_{\rm sf}^{\Phi_l,\Phi_e}\mathcal{P}_{\rm cf}^{\Phi_l}} \bigg] \nonumber\\
%&\leq&\mathbb{E}_{\Phi_l,\Phi_e}\bigg[\frac{\mathcal{P}_{\rm sf}^{\Phi_l,\Phi_e}}{\big(1-\mathcal{P}_{\rm cf}^{\Phi_l}\big)\big(1-\mathcal{P}_{\rm sf}^{\Phi_l,\Phi_e}\big)}\bigg] \nonumber\\
%\label{eqn:sop5}
%\end{eqnarray}

Therefore, we get the results in the theorem.

%%%%%%%%%%%%%%%%%%%%%%%%%%%%%%%%%%%%%%%%%%%%%%%%%%%%%%%%%%%%%%%%%%%%%%%%%%%%%%%%%%%%%%%%%%%%%%%%%%%%%%%%%%%%
\bibliographystyle{IEEEtran}
\bibliography{123}

\end{document}